\documentclass{amsart}
\usepackage{amssymb}
\usepackage{enumerate}
\usepackage{url}

\newtheorem{thm}{Theorem}
\newtheorem{prop}[thm]{Proposition}
\newtheorem{lem}[thm]{Lemma}
\newtheorem{cor}[thm]{Corollary}
\theoremstyle{definition}
\newtheorem{defn}[thm]{Definition}
\theoremstyle{remark}
\newtheorem{rem}[thm]{Remark}
\newtheorem{exam}[thm]{Example}

\numberwithin{equation}{section}
\numberwithin{thm}{section}

\newcounter{stepcounter}
\newcommand{\step}[1]{\refstepcounter{stepcounter} {\sl Step (\roman{stepcounter}) --- #1}}

\newcommand{\ZZ}{\mathbb{Z}}
\newcommand{\PP}{\mathcal{P}}
\newcommand{\RR}{\mathbb{R}}
\newcommand{\CC}{\mathbb{C}}
\newcommand{\FF}{\mathbb{F}}
\newcommand{\NN}{\mathbb{N}}
\newcommand{\QQ}{\mathbb{Q}}
\DeclareMathOperator{\Mint}{\mathsf{M}}
\newcommand{\Mss}{\mathop{\mathsf{M}_{\mathrm{SS}}}}
\DeclareMathOperator{\TT}{\mathsf{T}}
\newcommand{\TTlong}{\mathop{\mathsf{T}_{\mathrm{long}}}}
\newcommand{\TTshort}{\mathop{\mathsf{T}_{\mathrm{short}}}}
\newcommand{\MMshort}{\mathop{\mathsf{M}_{\mathrm{short}}}}
\newcommand{\MMbivariate}{\mathop{\mathsf{M}_{\mathrm{bivariate}}}}
\newcommand{\MMbivariateprime}{\mathop{\mathsf{M}'_{\mathrm{bivariate}}}}
\newcommand{\divides}{\mathrel|}
\newcommand{\ndivides}{\mathrel\nmid}
\newcommand{\norm}[1]{\Vert #1 \Vert}
\renewcommand{\leq}{\leqslant}
\renewcommand{\geq}{\geqslant}

\begin{document}

\title[Faster integer multiplication using short lattice vectors]
      {Faster integer multiplication \\ using short lattice vectors}

\author{David Harvey}
\address{School of Mathematics and Statistics, University of New South Wales, Sydney NSW 2052, Australia}
\curraddr{}
\email{d.harvey@unsw.edu.au}
\thanks{Harvey was supported by the Australian Research Council (grants DP150101689 and FT160100219).}

\author{Joris van der Hoeven}
\address{CNRS, Laboratoire d'informatique, \'Ecole polytechnique, 91128 Palaiseau, France}
\curraddr{}
\email{vdhoeven@lix.polytechnique.fr}
\thanks{}

\date{}

\dedicatory{}

\begin{abstract}
We prove that $n$-bit integers may be multiplied in
$O(n \log n \, 4^{\log^* n})$ bit operations.
This complexity bound had been achieved previously by several authors,
assuming various unproved number-theoretic hypotheses.
Our proof is unconditional, and depends in an essential way on
Minkowski's theorem concerning lattice vectors in symmetric convex sets.
\end{abstract}

\maketitle

\bibliographystyle{amsplain}

\section{Introduction}

Let $\Mint(n)$ denote the number of bit operations required to
multiply two $n$-bit integers,
where ``bit operations'' means the number of steps on a deterministic
Turing machine with a fixed, finite number of tapes \cite{Pap-complexity}
(our results also hold in the Boolean circuit model).
Let $\log^* x$ denote the iterated natural logarithm, i.e.,
$\log^* x := \min\,\{ j \in \NN : \log^{\circ j} x \leq 1\}$,
where $\log^{\circ j} x := \log \cdots \log x$ (iterated~$j$ times).
The main result of this paper is an algorithm achieving the following
bound.
\begin{thm}
\label{thm:main}
We have $\Mint(n) = O(n \log n \, 4^{\log^* n})$.
\end{thm}

The first complexity bound for $\Mint(n)$ of the form
$O(n \log n \, K^{\log^* n})$ was established by F\"urer
\cite{Fur-faster1, Fur-faster2}, for an unspecified constant $K > 1$.
His algorithm reduces a~multiplication of size $n$ to
many multiplications of size exponentially smaller than~$n$,
which are then handled recursively.
The number of recursion levels is $\log^* n + O(1)$,
and the constant~$K$ measures the ``expansion factor''
at each recursion level.

The first explicit value for $K$, namely $K = 8$,
was given by Harvey, van der Hoeven and Lecerf \cite{HvdHL-mul}.
Their method is somewhat different to F\"urer,
but still carries out an exponential size reduction at each recursion level.
One may think of the constant $K = 8$ as being built up of
three factors of $2$, each coming from a~different source.

The first factor of $2$ arises from the need to perform both
forward and inverse DFTs (discrete Fourier transforms) at each recursion level.
This is a feature common to all of the post-F\"urer algorithms,
suggesting that significantly new ideas will be needed to do any better than
$O(n \log n \, 2^{\log^* n})$.

The second factor of $2$ arises from coefficient growth:
a product of polynomials with $k$-bit integer coefficients
has coefficients with at least $2k$ bits.
This factor of $2$ also seems difficult to completely eliminate,
although Harvey and van der Hoeven have recently
made some progress \cite{HvdH-cyclotomic}:
they achieve $K = 4\sqrt 2 \approx 5.66$ by arranging that,
in effect, the coefficient growth only occurs at every
second recursion level.
This was the best known unconditional value of $K$ prior to the present paper.

The final factor of $2$ occurs because the algorithm works over $\CC$:
when multiplying complex coefficients with say $\beta$ significant bits,
the algorithm first computes a full $2\beta$-bit product,
and then truncates to $\beta$ bits.
More precisely, after splitting the $\beta$-bit inputs into
$m$ exponentially smaller chunks,
and encoding them into polynomials of degree $m$,
the algorithm must compute the full product of degree~$2m$,
even though essentially only $m$ coefficients are needed to resolve
$\beta$ significant bits of the product.
Again, this factor of $2$ has been the subject of a recent attack:
Harvey has shown \cite{Har-truncmul} that it is possible to work modulo
a polynomial of degree only $m$,
at the expense of increasing the working precision by a factor of $3/2$.
This leads to an integer multiplication algorithm achieving $K = 6$.

Another way of attacking this last factor of $2$ is to replace
the coefficient ring~$\CC$ by a finite ring $\ZZ/q\ZZ$
for a suitable integer $q$.
By choosing $q$ with some special structure,
it may become possible to convert a multiplication modulo $q$ directly
into a polynomial multiplication modulo
some polynomial of degree $m$, rather than~$2m$.
Three algorithms along these lines have been proposed.

First, Harvey, van der Hoeven and Lecerf suggested using
\emph{Mersenne primes}, i.e., primes of the form $q = 2^k - 1$,
where $k$ is itself prime \cite[\S9]{HvdHL-mul}.
They convert multiplication in $\ZZ/q\ZZ$ to multiplication in
$\ZZ[y]/(y^m - 1)$, where $m$ is a power of two.
Because $k$ is not divisible by $m$,
the process of splitting an element of $\ZZ/q\ZZ$ into $m$ chunks is
somewhat involved,
and depends on a variant of the Crandall--Fagin algorithm \cite{CF-DWT}.

Covanov and Thom\'e \cite{CT-zmult} later proposed using
\emph{generalised Fermat primes}, i.e., primes of the form $q = r^m + 1$,
where $m$ is a power of two and $r$ is a small even integer.
Here, multiplication in $\ZZ/q\ZZ$ is converted to multiplication in
$\ZZ[y]/(y^m + 1)$.
The splitting procedure consists of rewriting an element of $\ZZ/q\ZZ$
in base $r$, via fast radix-conversion algorithms.

Finally, Harvey and van der Hoeven \cite{HvdH-vanilla} proposed using
\emph{FFT primes}, i.e., primes of the form $q = a \cdot 2^k + 1$,
where $a$ is small.
They reduce multiplication in $\ZZ/q\ZZ$ to multiplication
in $\ZZ[y]/(y^m + a)$ via a straightforward splitting of the integers
into~$m$ chunks, where $m$ is a power of two.
Here the splitting process is trivial,
as $k$ may be chosen to be divisible by $m$.

These three algorithms all achieve $K = 4$, subject to plausible but
unproved conjectures on the distribution of the relevant primes.
Unfortunately, in all three cases, it is not even known that there are
infinitely many primes of the required form,
let alone that there exist a sufficiently high density of them
to satisfy the requirements of the algorithm.

The main technical novelty of the present paper is a splitting procedure
that works for an almost \emph{arbitrary} modulus $q$.
The core idea is to introduce an alternative representation for
elements of $\ZZ/q\ZZ$:
we represent them as expressions
$a_0 + a_1 \theta + \cdots + a_{m-1} \theta^{m-1}$,
where $\theta$ is some fixed $2m$-th root of unity in $\ZZ/q\ZZ$,
and where the $a_i$ are small integers, of size roughly $q^{1/m}$.
Essentially the only restriction on $q$ is that $\ZZ/q\ZZ$ must contain an
appropriate $2m$-th root of unity.
We will see that Linnik's theorem is strong enough to construct
suitable such moduli $q$.

In Section \ref{sec:theta} we show that the cost of multiplication
in this representation is only a constant factor worse than for
the usual representation.
The key ingredient is Minkowski's theorem on lattice vectors
in symmetric convex sets.
We also give algorithms for converting between this representation and
the standard representation.
The conversions are not as fast as one might hope ---
in particular, we do not know how to carry them out in quasilinear time ---
but surprisingly this turns out not to affect the overall complexity,
because in the main multiplication algorithm we perform
the conversions only infrequently.

Then in Sections \ref{sec:main} and \ref{sec:top}
we prove Theorem \ref{thm:main},
using an algorithm that is structurally very similar to \cite{HvdH-vanilla}.
We make no attempt to minimise the implied big-$O$ constant in
Theorem \ref{thm:main};
our goal is to give the simplest possible proof of the asymptotic bound,
without any regard for questions of practicality.

An interesting question is whether it is possible to combine the techniques
of the present paper with those of
\cite{HvdH-cyclotomic} to obtain an algorithm
achieving $K = 2 \sqrt 2 \approx 2.83$.
Our attempts in this direction have so far been unsuccessful.
One might also ask if the techniques of this paper can be transferred
to the case of multiplication of polynomials of high degree in $\FF_p[x]$.
However, this is not so interesting, because an unconditional proof of
the bound corresponding to $K = 4$ in the polynomial case is already known
\cite{HvdH-cyclotomic}.

Throughout the paper we use the following notation.
We write $\lg n := \lceil \log_2 n \rceil$ for $n \geq 2$,
and for convenience put $\lg 1 := 1$.
We define $\Mss(n) = C n \lg n \lg \lg n$,
where $C > 0$ is some constant so that the Sch\"onhage--Strassen algorithm
multiplies $n$-bit integers in at most $\Mss(n)$ bit operations
\cite{SS-multiply}.
This function satisfies $n \Mss(m) \leq \Mss(nm)$ for any $n, m \geq 1$, 
and also $\Mss(dm) = O(\Mss(m))$ for fixed~$d$.
An $n$-bit integer may be divided by $m$-bit integer,
producing quotient and remainder, in time $O(\Mss(\max(n, m)))$
\cite[Ch.~9]{vzGG-compalg3}.
We may transpose an $n \times m$ array of objects of bit size $b$ in
$O(b n m \lg \min(n, m))$ bit operations \cite[Appendix]{BGS-recurrences}.
Finally, we occasionally use Xylouris's refinement of Linnik's theorem
\cite{Xyl-linnik},
which states that for any relatively prime positive integers $a$ and $n$,
the least prime in the arithmetic progression $p = a \pmod n$
satisfies $p = O(n^{5.2})$.

\section{$\theta$-representations}
\label{sec:theta}

Throughout this section,
fix an integer $q \geq 2$ and a power of two $m$ such that
\begin{equation}
\label{eq:m-bound}
 m \leq \frac{\log_2 q}{(\lg \lg q)^2},
 \qquad \text{or equivalently,} \qquad
 q^{1/m} \geq 2^{(\lg \lg q)^2},
\end{equation}
and assume we are given some $\theta \in \ZZ/q\ZZ$ such that $\theta^m = -1$.

For a polynomial
$F = F_0 + F_1 y + \cdots + F_{m-1} y^{m-1} \in \ZZ[y]/(y^m + 1)$,
define $\norm{F} := \max_i |F_i|$.
This norm satisfies $\norm{FG} \leq m \norm F \norm G$
for any $F, G \in \ZZ[y]/(y^m + 1)$.
\begin{defn}
Let $u \in \ZZ/q\ZZ$.
A \emph{$\theta$-representation for $u$} is a polynomial
$U \in \ZZ[y]/(y^m + 1)$ such that $U(\theta) = u \pmod q$ and
$\norm{U} \leq m q^{1/m}$.
\end{defn}
\begin{exam}
\label{exam:theta}
Let $m = 4$ and
\begin{align*}
      q & = 3141592653589793238462833, \\
 \theta & = 2542533431566904450922735 \pmod q, \\
      u & = 2718281828459045235360288 \pmod q.
\end{align*}
The coefficients in a $\theta$-representation must not exceed
$m q^{1/m} \approx 5325341.46$.
Two examples of $\theta$-representations for $u$ are
\begin{align*} 
 U(y) & = -3366162y^3 + 951670y^2 - 5013490y - 3202352, \\
 U(y) & = -4133936y^3 + 1849981y^2 - 5192184y + 1317423.
\end{align*}
\end{exam}

By \eqref{eq:m-bound}, the number of bits required to store $U(y)$ is at most
\begin{equation*}
 m \big(\log_2(m q^{1/m}) + O(1)\big) = \lg q + O(m \lg m) = \left(1 + \frac{O(1)}{\lg \lg q}\right) \lg q,
\end{equation*}
so a $\theta$-representation incurs very little overhead in space
compared to the standard representation by an integer in the interval
$0 \leq x < q$.

Our main tool for working with $\theta$-representations is the
\emph{reduction algorithm} stated in Lemma \ref{lem:reduction} below.
Given a polynomial $F \in \ZZ[y]/(y^m + 1)$,
whose coefficients are up to about twice as large as allowed
in a $\theta$-representation,
the reduction algorithm computes a $\theta$-representation for $F(\theta)$
(up to a certain scaling factor, discussed further below).
The basic idea of the algorithm is to precompute a nonzero polynomial $P(y)$
such that $P(\theta) = 0 \pmod q$,
and then to subtract an appropriate multiple of $P(y)$ from $F(y)$
to make the coefficients small.

After developing the reduction algorithm, we are able to give algorithms for
basic arithmetic on elements of $\ZZ/q\ZZ$ given in $\theta$-representation
(Proposition \ref{prop:arithmetic}),
a more general reduction algorithm for inputs of arbitrary size
(Proposition \ref{prop:reduction}),
and algorithms for converting between standard and $\theta$-representation
(Proposition \ref{prop:convert-to-theta} and
Proposition \ref{prop:convert-to-standard}).

We begin with two results that generate certain precomputed data necessary
for the main reduction step.
\begin{lem}
\label{lem:compute-P}
In $q^{1+o(1)}$ bit operations,
we may compute a nonzero polynomial $P \in \ZZ[y]/(y^m + 1)$ such that
$P(\theta) = 0 \pmod q$ and $\norm{P} \leq q^{1/m}$.
\end{lem}
\begin{proof}
We first establish existence of a suitable $P(y)$.
Let $\overline{\theta^i}$ denote a lift of $\theta^i$ to $\ZZ$,
and consider the lattice $\Lambda \subset \ZZ^m$ spanned by the
rows of the $m \times m$ integer matrix
\begin{equation*}
   A =
   \begin{pmatrix}
      q                        & 0 & 0 &        & 0 \\
      -\overline{\theta}       & 1 & 0 & \cdots & 0 \\
      -\overline{\theta^2}     & 0 & 1 &        & 0 \\
      \vdots                   &   &   & \ddots &   \\
      -\overline{\theta^{m-1}} & 0 & 0 & \cdots & 1
   \end{pmatrix}
\end{equation*}
Every vector $(a_0, \ldots, a_{m-1}) \in \Lambda$ satisfies the equation
$a_0 + \cdots + a_{m-1} \theta^{m-1} = 0 \pmod q$.
The volume of the fundamental domain of $\Lambda$ is $\det A = q$.
The volume of the closed convex symmetric set
$\Sigma := \{|a_i| \leq q^{1/m}\} \subset \RR^m$ is $(2q^{1/m})^m = 2^m q$,
so by Minkowski's theorem (see for example \cite[Ch.~V, Thm.~3]{Lan-ANT}),
there exists a nonzero vector $(a_0, \ldots, a_{m-1})$
in $\Lambda \cap \Sigma$.
The corresponding polynomial $P(y) := a_0 + \cdots + a_{m-1} y^{m-1}$
then has the desired properties.

To actually compute $P(y)$, we simply perform a brute-force search.
By \eqref{eq:m-bound} there are at most
$(2q^{1/m} + 1)^m \leq (3 q^{1/m})^m = 3^m q < q^{1+o(1)}$ candidates to test.
Enumerating them in lexicographical order,
we can easily evaluate $P(\theta) \pmod q$ in an average of
$O(\lg q)$ bit operations per candidate.
\end{proof}
\begin{exam}
Continuing Example \ref{exam:theta}, the coefficients of $P(y)$ must not
exceed $q^{1/m} \approx 1331335.36$.
A suitable polynomial $P(y)$ is given by
 \[ P(y) = -394297y^3 - 927319y^2 + 1136523y - 292956. \]
\end{exam}

\begin{rem}
The computation of $P(y)$ is closely related to the problem of finding an
element of small norm in the ideal of the ring $\ZZ[\zeta_{2m}]$ generated by
$q$ and $\zeta_{2m} - \overline\theta$,
where $\zeta_{2m}$ denotes a primitive $2m$-th root of unity.
\end{rem}

\begin{rem}
The poor exponential-time complexity of Lemma \ref{lem:compute-P}
can probably be improved,
by taking advantage of more sophisticated lattice reduction or
shortest vector algorithms,
but we were not easily able to extract a suitable result from the literature.
For example, LLL is not guaranteed to produce a short enough vector
\cite{LLL-factoring},
and the Micciancio--Voulgaris exact shortest vector algorithm \cite{MV-voronoi}
solves the problem for the Euclidean norm rather than the uniform norm.
In any case, this has no effect on our main result.
\end{rem}

\begin{lem}
\label{lem:compute-J}
Assume that $P(y)$ has been precomputed as in Lemma \ref{lem:compute-P}.
Let $r$ be the smallest prime exceeding $2 m^2 q^{1/m}$ such that
$r \ndivides q$ and such that $P(y)$ is invertible in
$(\ZZ/r\ZZ)[y]/(y^m + 1)$.
Then $r = O(m^2 q^{1/m})$, and in $q^{1+o(1)}$ bit operations we may
compute $r$ and a polynomial $J \in \ZZ[y]/(y^m + 1)$ such that
$J(y) P(y) = 1 \pmod r$ and $\norm J \leq r$.
\end{lem}
\begin{proof}
Let $R \in \ZZ$ be the resultant of $P(y)$
(regarded as a polynomial in $\ZZ[y]$) and $y^m + 1$.
The primes $r$ dividing $R$ are exactly the primes for which
$P(y)$ fails to be invertible in $(\ZZ/r\ZZ)[y]/(y^m + 1)$.
Therefore our goal is to find a prime $r > 2m^2 q^{1/m}$
such that $r \ndivides Rq$.

Since $m$ is a power of two,
$y^m + 1$ is a cyclotomic polynomial and hence irreducible in $\QQ[y]$.
Thus $y^m + 1$ and $P(y)$ have no common factor, and so $R \neq 0$.
Also, we have $R = \prod_\alpha P(\alpha)$ where
$\alpha$ runs over the complex roots of $y^m + 1$.
These roots all lie on the unit circle,
so $|P(\alpha)| \leq m \norm P \leq m q^{1/m}$,
and hence by \eqref{eq:m-bound} we obtain
$|Rq| \leq (m q^{1/m})^m q = m^m q^2 < q^3$.

On the other hand, the prime number theorem
(in the form $\sum_{p < x} \log p \sim x$,
see for example \cite[\S4.3]{Apo-analytic})
implies that there exists
an absolute constant $C > 2$ such that for any $x \geq 1$ we have
$\sum_{2x \leq p \leq Cx} \log_2 p \geq 3x$ (sum taken over primes).
Taking $x := m^2 q^{1/m}$, by \eqref{eq:m-bound} again we get
\begin{equation*}
  \sum_{2m^2 q^{1/m} \leq p \leq Cm^2 q^{1/m}} \log_2 p
    \geq 3 m^2 q^{1/m}
    \geq 3 \cdot 2^{(\lg \lg q)^2}
    \geq 3 \lg q
    \geq \log_2(q^3).
\end{equation*}
In particular, there must be at least one prime in the interval
$2 m^2 q^{1/m} \leq r \leq C m^2 q^{1/m}$ that does not divide $Rq$.

To find the smallest such $r$, we first make a list of all primes up to
$C m^2 q^{1/m}$ in $(C m^2 q^{1/m})^{1+o(1)} < q^{1+o(1)}$ bit operations.
Then for each prime $r$ between $2 m^2 q^{1/m}$ and $C m^2 q^{1/m}$,
we check whether $r$ divides $q$ in $(\lg q)^{1+o(1)}$ bit operations,
and attempt to invert $P(y)$ in $(\ZZ/r\ZZ)[y]/(y^m + 1)$ in
$(m \lg r)^{1+o(1)} = (\lg q)^{1+o(1)}$ bit operations
\cite[Ch.~11]{vzGG-compalg3}.
\end{proof}

\begin{exam}
Continuing Example \ref{exam:theta}, we have $r = 42602761$ and
 \[ J(y) = 17106162y^3 + 6504907y^2 + 30962874y + 8514380. \]
\end{exam}
Now we come to the main step of the reduction algorithm, which is inspired by
Montgomery's method for modular reduction \cite{Mon-modular}.
\begin{lem}
\label{lem:reduction}
Assume that $P(y)$, $r$ and $J(y)$ have been precomputed as in
Lemmas~\ref{lem:compute-P} and~\ref{lem:compute-J}.
Given as input $F \in \ZZ[y]/(y^m + 1)$ with $\norm F \leq m^3 (q^{1/m})^2$,
we may compute a $\theta$-representation for $F(\theta)/r \pmod q$
in $O(\Mss(\lg q))$ bit operations.
\end{lem}
\begin{proof}
We first compute the ``quotient'' $Q := FJ \pmod r$,
normalised so that $\norm Q \leq r/2$.
This is done by means of Kronecker substitution \cite[Ch.~8]{vzGG-compalg3},
i.e., we pack the polynomials $F(y)$ and $J(y)$ into integers,
multiply the integers, unpack the result,
and reduce the result modulo $y^m + 1$ and modulo $r$.
The packed integers have at most $m (\lg \norm F + \lg r + \lg m)$ bits,
where the $\lg m$ term accounts for coefficient growth in $\ZZ[y]$.
By \eqref{eq:m-bound} and Lemma \ref{lem:compute-J},
this simplifies to $O(\lg q)$ bits, so the integer multiplication step costs
$O(\Mss(\lg q))$ bit operations.
This bound also covers the cost of the reductions modulo $r$.

Next we compute the product $QP$, again using Kronecker substitution,
at a cost of $O(\Mss(\lg q))$ bit operations.
Since $\norm Q \leq r/2$ and $\norm P \leq q^{1/m}$,
we have $\norm{QP} \leq \frac12 r m q^{1/m}$.

By construction of $J$ we have $QP = F \pmod r$.
In particular, all the coefficients of $F - QP \in \ZZ[y]/(y^m + 1)$
are divisible by $r$.
The last step is to compute the ``remainder'' $G := (F - QP)/r$;
again, this step costs $O(\Mss(\lg q))$ bit operations.
Since $r \geq 2 m^2 q^{1/m}$, we have
\begin{equation*}
  \norm{G} \leq \frac{\norm{F}}{r} + \frac{\norm{QP}}{r}
           \leq \frac{m^3 (q^{1/m})^2}{2 m^2 q^{1/m}} + \frac{m q^{1/m}}{2}
           \leq m q^{1/m}.
\end{equation*}
Finally, since $P(\theta) = 0 \pmod q$,
and all arithmetic throughout the algorithm has been performed modulo $y^m + 1$,
we see that $G(\theta) = F(\theta)/r \pmod q$.
\end{proof}

Using the above reduction algorithm,
we may give preliminary addition and multiplication algorithms for
elements of $\ZZ/q\ZZ$ in $\theta$-representation.
\begin{lem}
\label{lem:arithmetic}
Assume that $P(y)$, $r$ and $J(y)$ have been precomputed as in
Lemmas~\ref{lem:compute-P} and~\ref{lem:compute-J}.
Given as input $\theta$-representations for $u, v \in \ZZ/q\ZZ$,
we may compute $\theta$-representations for $uv/r$ and $(u \pm v)/r$
in $O(\Mss(\lg q))$ bit operations.
\end{lem}
\begin{proof}
Let the $\theta$-representations be given by $U, V \in \ZZ[y]/(y^m + 1)$.
We may compute $F_* := UV$ in $\ZZ[y]/(y^m + 1)$ using Kronecker substitution
in $O(\Mss(\lg q))$ bit operations,
and $F_{\pm} := U \pm V$ in $O(\lg q)$ bit operations.
Note that $\norm{F_*} \leq m \norm{U} \norm{V} \leq m^3 (q^{1/m})^2$,
and $\norm{F_\pm} \leq \norm{U} + \norm{V}
\leq 2 m q^{1/m} \leq m^3 (q^{1/m})^2$,
so we may apply Lemma \ref{lem:reduction} to obtain the desired
$\theta$-representations.
\end{proof}

\begin{exam}
Continuing Example \ref{exam:theta}, we walk through an example of computing
a product of elements in $\theta$-representation.
Let
\begin{align*}
  u & = 1414213562373095048801689 \pmod q, \\
  v & = 1732050807568877293527447 \pmod q.
\end{align*}
Suppose we are given as input the $\theta$-representations
\begin{align*}
  U(y) & = 3740635y^3 + 3692532y^2 - 3089740y + 4285386, \\
  V(y) & = 4629959y^3 - 4018180y^2 - 2839272y - 3075767.
\end{align*}
We first compute the product of $U(y)$ and $V(y)$ modulo $y^m + 1$:
\begin{multline*}
  F(y) = U(y)V(y) = 10266868543625y^3 - 37123194804209y^2 \\
            - 4729783170300y + 26582459129078.
\end{multline*}
We multiply $F(y)$ by $J(y)$ and reduce modulo $r$ to obtain the quotient
 \[ Q(y) = 3932274y^3 - 14729381y^2 + 20464841y - 11934644. \]
Then the remainder
 \[ (F(y) - P(y) Q(y))/r = 995963y^3 - 1814782y^2 + 398819y + 777998 \]
is a $\theta$-representation for $uv/r \pmod q$.
\end{exam}

The following precomputation will assist in eliminating the spurious $1/r$
factor appearing in Lemmas~\ref{lem:reduction} and~\ref{lem:arithmetic}.
\begin{lem}
\label{lem:compute-D}
Assume that $P(y)$, $r$ and $J(y)$ have been precomputed as in
Lemmas~\ref{lem:compute-P} and~\ref{lem:compute-J}.
In $q^{1+o(1)}$ bit operations, we may compute a polynomial
$D \in \ZZ[y]/(y^m + 1)$ such that $\norm{D} \leq m q^{1/m}$ and
$D(\theta) = r^2 \pmod q$.
\end{lem}
\begin{proof}
We may easily compute the totient function $\varphi(q)$ in
$q^{1+o(1)}$ bit operations, by first factoring $q$.
Since $(r,q) = 1$, we have $r^{-(\varphi(q) - 2)} = r^2 \pmod q$.
Repeatedly using the identity $r^{-i-1} = (r^{-i} \cdot 1)/r$,
we may compute $\theta$-representations for
$r^{-1}, r^{-2}, \ldots, r^{-(\varphi(q)-2)}$
by successively applying Lemma \ref{lem:arithmetic}.
\end{proof}
\begin{rem}
Assuming the factorisation of $q$ is known
(which will always be the case in the application in Section \ref{sec:main}),
the complexity of Lemma \ref{lem:compute-D} may be improved to
$O(\Mss(\lg q) \lg q)$ bit operations by using a modified
``repeated squaring'' algorithm.
\end{rem}
\begin{exam}
Continuing Example \ref{exam:theta}, we may take
 \[ D(y) = -1918607y^3 - 3680082y^2 + 2036309y - 270537. \] 
\end{exam}

Henceforth we write $\PP(q, m, \theta)$ for the tuple $(P(y), r, J(y), D(y))$
of precomputed data generated by Lemmas~\ref{lem:compute-P},
\ref{lem:compute-J}, and~\ref{lem:compute-D}.
Given $q$, $m$ and $\theta$ as input, the above results show that we may
compute $\PP(q, m, \theta)$ in $q^{1+o(1)}$ bit operations.
With these precomputations out of the way,
we may state complexity bounds for the main operations on
$\theta$-representations.

\begin{prop}
\label{prop:arithmetic}
Assume that $\PP(q, m, \theta)$ has been precomputed.
Given as input $\theta$-representations for $u, v \in \ZZ/q\ZZ$,
we may compute $\theta$-representations for $uv$ and $u \pm v$ in
$O(\Mss(\lg q))$ bit operations.
\end{prop}
\begin{proof}
For the product, we first use Lemma \ref{lem:arithmetic} to compute a
$\theta$-representation for $uv/r \pmod q$,
and then we use Lemma \ref{lem:arithmetic} again to multiply by $D(y)$,
to obtain a $\theta$-representation for $(uv/r)(r^2)/r = uv \pmod q$.
The sum and difference are handled similarly.
\end{proof}
\begin{rem}
We suspect that the complexity bound for $u \pm v$ can be improved
to $O(\lg q)$, but we do not currently know how to achieve this.
This question seems closely related to Remark \ref{rem:archimedean} below.
\end{rem}

\begin{prop}
\label{prop:reduction}
Assume that $\PP(q, m, \theta)$ has been precomputed.
Given as input a polynomial $F \in \ZZ[y]/(y^m + 1)$
(with no restriction on $\norm F$),
we may compute a~$\theta$-representation for $F(\theta) \pmod q$ in time
$O(\lceil m \lg \norm F / \lg q \rceil \Mss(\lg q))$.
\end{prop}
\begin{proof}
Let $b := \lg \, \lceil q^{1/m} \rceil$ and
$n := \lceil 2 m \lg \norm F / \lg q \rceil$, so that
\begin{equation*}
  2^{nb} \geq (q^{1/m})^n
         \geq (q^{1/m})^{2 m \lg \norm F / \lg q}
         =    2^{\lg \norm F (2 \log_2 q / \lg q)}
         \geq 2^{\lg \norm F}.
\end{equation*}
We may therefore decompose the coefficients of $F$ into $n$ chunks of~$b$ bits,
i.e., we may compute polynomials $F_0, \ldots, F_{n-1} \in \ZZ[y]/(y^m + 1)$ 
such that $F = F_0 + 2^b F_1 + \cdots + 2^{(n-1)b} F_{n-1}$
and $\norm{F_i} \leq 2^b \leq 2 q^{1/m}$.
(This step implicitly requires an array transposition of cost
$O(bmn \lg m) = O(n \lg q \lg \lg q)$.) 
Now we use Proposition~\ref{prop:arithmetic} repeatedly to compute a
$\theta$-representation for $F$ via Horner's rule,
i.e., first we compute a $\theta$-representation for $2^b F_{n-1} + F_{n-2}$,
then for $2^b ( 2^b F_{n-1} + F_{n-2}) + F_{n-3}$, and so on.
\end{proof}

\begin{prop}
\label{prop:convert-to-theta}
Assume that $\PP(q, m, \theta)$ has been precomputed.
Given as input an element $u \in \ZZ/q\ZZ$ in standard representation,
we may compute a~$\theta$-representation for $u$ in
$O(m \Mss(\lg q))$ bit operations.
\end{prop}
\begin{proof}
Simply apply Proposition \ref{prop:reduction} to the
constant polynomial $F(y) = u$, noting that $\norm F \leq q$.
\end{proof}

\begin{rem}
A corollary of Proposition \ref{prop:convert-to-theta} is that every
$u \in \ZZ/q\ZZ$ admits a $\theta$-representation.
It would be interesting to have a direct proof of this fact that does
not rely on the reduction algorithm.
A related question is whether it is possible to tighten the bound in the
definition of $\theta$-representation from $m q^{1/m}$ to $q^{1/m}$,
or even $\frac12 q^{1/m}$.
We do not know whether such a representation exists for all $u \in \ZZ/q\ZZ$.
\end{rem}

\begin{prop}
\label{prop:convert-to-standard}
Given as input an element $u \in \ZZ/q\ZZ$ in $\theta$-representation,
we may compute the standard representation for $u$ in
$O(m \Mss(\lg q))$ bit operations.
\end{prop}
\begin{proof}
Let $U \in \ZZ[y]/(y^m + 1)$ be the input polynomial.
The problem amounts to evaluating $U(\theta)$ in $\ZZ/q\ZZ$.
Again we may simply use Horner's rule.
\end{proof}

\begin{rem}
In both Proposition \ref{prop:convert-to-theta} and
Proposition \ref{prop:convert-to-standard},
the input and output have bit size $O(\lg q)$,
but the complexity bounds given are not quasilinear in $\lg q$.
It is possible to improve on the stated bounds,
but we do not know a quasilinear time algorithm for the conversion
in either direction.
\end{rem}

\begin{rem}
\label{rem:archimedean}
In the reduction algorithm, the reader may wonder why we go to the trouble of
introducing the auxiliary prime $r$.
Why not simply precompute an approximation to a \emph{real} inverse for $P(y)$,
i.e., the inverse in $\RR[y]/(y^m + 1)$,
and use this to clear out the \emph{high-order bits} of
each coefficient of the dividend?
In other words, why not replace the Montgomery-style division
with the more natural Barrett-style division \cite{Bar-RSA}?

The reason is that we cannot prove tight enough bounds on the size of the
coefficients of this inverse:
it is conceivable that $P(y)$ might accidentally take on a
very small value near one of the complex roots of $y^m + 1$, or equivalently,
that the resultant $R$ in the proof of Lemma \ref{lem:compute-J}
might be unusually small.
For the same reason, we cannot use a more traditional
$2$-adic Montgomery inverse to clear out the low-order bits of the dividend,
because again $P(y)$ may take a $2$-adically small value near one of the
$2$-adic roots of $y^m + 1$, or equivalently,
the resultant $R$ might be divisible by an unusually large power of $2$.
\end{rem}

\section{Integer multiplication: the recursive step}
\label{sec:main}

In this section we present a recursive routine \textsc{Transform}
with the following interface.
It takes as input a (sufficiently large) power-of-two transform length $L$,
a prime $p = 1 \pmod L$, a prime power $q = p^\alpha$ such that
\begin{equation}
\label{eq:lgq-bound}
 \lg L \leq \lg q \leq 3 \lg L \lg \lg L,
\end{equation}
a principal $L$-th root of unity $\zeta \in \ZZ/q\ZZ$
(i.e., an $L$-th root of unity whose reduction modulo $p$ is
a primitive $L$-th root of unity in the field $\ZZ/p\ZZ$),
certain precomputed data depending on $L$, $q$ and $\zeta$ (see below),
and a polynomial $F \in (\ZZ/q\ZZ)[x]/(x^L - 1)$.
Its output is the DFT of $F$ with respect to $\zeta$, that is, the vector
 \[ \hat F := (F(1), F(\zeta), \ldots, F(\zeta^{L-1})) \in (\ZZ/q\ZZ)^L. \]
The coefficients of both $F$ and $\hat F$ are given in standard representation.

The precomputed data consists of the tuple $\PP(q, m, \theta)$
defined in Section \ref{sec:theta},
where~$m$ and $\theta$ are defined as follows.

First, \eqref{eq:lgq-bound} implies that
$\lg q \geq (\lg \lg L)^2 \lg \lg \lg L$ for large $L$,
so we may take $m$ to be the unique power of two lying in the interval
\begin{equation}
\label{eq:m-interval}
 \frac{\lg q}{(\lg \lg L)^2 \lg \lg \lg L}
    \leq m
    < \frac{2 \lg q}{(\lg \lg L)^2 \lg \lg \lg L}.
\end{equation}
Observe that \eqref{eq:m-bound} is certainly satisfied for this choice of $m$
(for large enough $L$),
as \eqref{eq:lgq-bound} implies that $\lg \lg L \sim \lg \lg q$.

Next, note that $2m \divides L$, because \eqref{eq:lgq-bound} and
\eqref{eq:m-interval} imply that $m = o(\lg L) = o(L)$;
therefore we may take $\theta := \zeta^{L/2m}$,
so that $\theta^m = \zeta^{L/2} = -1$.

We remark that the role of the parameter $\alpha$ is to give us enough
control over the bit size of $q$,
to compensate for the fact that Linnik's theorem does not give us
sufficiently fine control over the bit size of $p$
(see Lemma \ref{lem:qprime-bound}).

Our implementation of \textsc{Transform} uses one of two algorithms,
depending on the size of $L$.
If $L$ is below some threshold, say $L_0$,
then it uses any convenient base-case algorithm.
Above this threshold, it reduces the given DFT problem
to a~collection of exponentially smaller DFTs of the same type,
via a series of reductions that may be summarised as follows.

\begin{enumerate}[(i)]
\item Use the conversion algorithms from Section \ref{sec:theta}
to reduce to a transform over $\ZZ/q\ZZ$ where the
input and output coefficients are given in $\theta$-representation.
(During steps (ii) and (iii) below,
all elements of $\ZZ/q\ZZ$ are stored and manipulated
entirely in $\theta$-representation.)

\item Reduce the ``long'' transform of length $L$ over $\ZZ/q\ZZ$ to many
``short'' transforms of exponentially small length
$S := 2^{(\lg \lg L)^2}$ over $\ZZ/q\ZZ$,
via the Cooley--Tukey decomposition.

\item Reduce each short transform from step (ii)
to a product in $(\ZZ/q\ZZ)[x]/(x^S - 1)$,
i.e., a cyclic convolution of length $S$, using Bluestein's algorithm.

\item Use the definition of $\theta$-representation to
reinterpret each product from (iii) as a~product in
$\ZZ[x,y]/(x^S - 1, y^m + 1)$,
where the coefficients in $\ZZ$ are exponentially smaller than the
original coefficients in $\ZZ/q\ZZ$.

\item Embed each product from (iv) into $(\ZZ/q'\ZZ)[x,y]/(x^S - 1, y^m  + 1)$
for a suitable prime power $q'$ that is exponentially smaller than $q$,
and large enough to resolve the coefficients of the products over $\ZZ$.

\item Reduce each product from (v) to a collection of
forward and inverse DFTs of length $S$ over $\ZZ/q'\ZZ$, and recurse.
\end{enumerate}

The structure of this algorithm is very similar to that of \cite{HvdH-vanilla}.
The main difference is that it is not necessary to explicitly
split the coefficients into chunks in step~(iv);
this happens automatically as a consequence of storing the coefficients
in $\theta$-representation.
In effect, the splitting (and reassembling) work has been shunted
into the conversions in step (i).

We now consider each of the above steps in more detail.
We write $\TT(L,q)$ for the running time of \textsc{Transform}.
We always assume that $L_0$ is increased
whenever necessary to accommodate statements that hold only for large $L$.

\medskip
\step{convert between representations.}
Let $\TTlong(L, q)$ denote the time required to compute a DFT of length $L$
over $\ZZ/q\ZZ$ with respect to $\zeta$,
assuming that the coefficients of the input $F$ and the output $\hat F$
are given in $\theta$-representation,
and assuming that $\PP(q, m, \theta)$ is known.
\begin{lem}
\label{lem:step-convert}
We have $\TT(L, q) < \TTlong(L, q) + O(L \lg L \lg q)$.
\end{lem}
\begin{proof}
We first convert $F$ from standard to $\theta$-representation using Proposition
\ref{prop:convert-to-theta};
we then compute $\hat F$ from $F$ (working entirely in $\theta$-representation);
at the end, we convert $\hat F$ back to standard representation
using Proposition \ref{prop:convert-to-standard}.
By~\eqref{eq:lgq-bound} and~\eqref{eq:m-interval},
the total cost of the conversions is
\begin{align*}
 O(L m \Mss(\lg q))
  &= O\left(L \frac{\lg q}{(\lg \lg L)^2 \lg \lg \lg L} \lg q \lg \lg q \lg \lg \lg q \right) \\
  &= O\left(L \frac{\lg L \lg \lg L}{(\lg \lg L)^2 \lg \lg \lg L} \lg q \lg \lg L \lg \lg \lg L \right)\\
  &= O(L \lg L \lg q).   \qedhere
\end{align*}
\end{proof}

Henceforth all elements of $\ZZ/q\ZZ$ are assumed to be
stored in $\theta$-representation,
and we will always use Proposition \ref{prop:arithmetic}
to perform arithmetic operations on such elements
in $O(\Mss(\lg q))$ bit operations.

\medskip
\step{reduce to short DFTs.}
Let $S := 2^{(\lg \lg L)^2}$.
Given as input polynomials $F_1, \ldots, F_{L/S} \in (\ZZ/q\ZZ)[x]/(x^S - 1)$
(presented sequentially on tape),
let $\TTshort(L, q)$ denote the time required to compute the transforms
$\hat F_1, \ldots, \hat F_{L/S} \in (\ZZ/q\ZZ)^S$ with respect to
the principal $S$-th root of unity $\omega := \zeta^{L/S}$.
(Here and below, we continue to assume that $\PP(q,m,\theta)$ is known.).
\begin{lem}
\label{lem:step-short}
We have $\TTlong(L,q) < \frac{\lg L}{(\lg \lg L)^2} \TTshort(L,q) + O(L \lg L \lg q)$.
\end{lem}
\begin{proof}
Let $d := \lfloor \lg L / \lg S \rfloor$,
so that $\lg L = d \lg S + d'$ where $0 \leq d' < \lg S$.
Applying the Cooley--Tukey method \cite{CT-fft}
to the factorisation $L = S^d 2^{d'}$,
the given transform of length $L$ may be decomposed into~$d$ layers,
each consisting of $L/S$ transforms of length~$S$ (with respect to $\omega$),
followed by $d'$ layers, each consisting of $L/2$ transforms of length~$2$.
Between each of these layers,
we must perform $O(L)$ multiplications by ``twiddle factors'' in $\ZZ/q\ZZ$,
which are given by certain powers of $\zeta$.
(For further details of the Cooley--Tukey decomposition,
see for example \cite[\S2.3]{HvdHL-mul}.)

The total cost of the twiddle factor multiplications,
including the cost of computing the twiddle factors themselves, is
\begin{align*}
 O((d + d') L \Mss(\lg q))
  &= O\left( \left(\frac{\lg L}{(\lg \lg L)^2} + (\lg \lg L)^2\right) L \lg q \lg \lg q \lg \lg \lg q \right) \\
  &= O\left( \frac{\lg L}{(\lg \lg L)^2} L \lg q \lg \lg L \lg \lg \lg L\right) = O(L \lg L \lg q).
\end{align*}
This bound also covers the cost of the length~$2$ transforms (`butterflies'),
each of which requires one addition and one subtraction in $\ZZ/q\ZZ$.

In the Turing model, we must also account for the cost of rearranging data
so that the inputs for each layer of short DFTs are stored sequentially on tape.
The cost per layer is $O(L \lg S \lg q)$ bit operations,
so $O(L \lg L \lg q)$ altogether
(see~\cite[\S2.3]{HvdHL-mul} for further details).
\end{proof}

\step{reduce to short convolutions.}
Given polynomials $G_1, \ldots, G_{L/S}, H \in (\ZZ/q\ZZ)[x]/(x^S - 1)$ as input,
let $\MMshort(L,q)$ denote the time required to compute the products
$G_1 H, \ldots, G_{L/S} H$.
\begin{lem}
\label{lem:step-convolution}
We have $\TTshort(L,q) < \MMshort(L,q) + O(L (\lg \lg L)^2 \lg q)$.
\end{lem}
\begin{proof}
We use Bluestein's method \cite{Blu-dft},
which reduces the the problem of computing the
DFT of $F \in (\ZZ/q\ZZ)[x]/(x^S - 1)$
to the problem of computing the product of certain polynomials
$G, H \in (\ZZ/q\ZZ)[x]/(x^S - 1)$,
plus $O(S)$ auxiliary multiplications in $\ZZ/q\ZZ$
(for further details see \cite[\S2.5]{HvdHL-mul}).
Here $G$ depends on $F$ and~$\zeta$, but $H$ depends only on $\zeta$.
The total cost of the auxiliary multiplications is
  \[ O((L/S) S \Mss(\lg q)) = O(L \lg q \lg \lg q \lg \lg \lg q) = O(L (\lg \lg L)^2 \lg q). \qedhere \]
\end{proof}

\step{reduce to bivariate products over $\ZZ$.}
Given as input polynomials
$\tilde G_1, \ldots, \tilde G_{L/S}, \tilde H \in \ZZ[x,y]/(x^S - 1, y^m + 1)$,
all whose of coefficients are bounded in absolute value by $m q^{1/m}$,
let $\MMbivariate(L,q)$ denote the cost of computing the products
$\tilde G_1 \tilde H, \ldots, \tilde G_{L/S} \tilde H$.
\begin{lem}
\label{lem:step-bivariate}
We have $\MMshort(L,q) < \MMbivariate(L,q) + O(L (\lg \lg L)^2 \lg q)$.
\end{lem}
\begin{proof}
We are given as input polynomials
$G_1, \ldots, G_{L/S}, H \in (\ZZ/q\ZZ)[x]/(x^S - 1)$.
Since their coefficients are given in $\theta$-representation,
we may immediately reinterpret them as polynomials 
$\tilde G_1, \ldots, \tilde G_{L/S}, \tilde H \in \ZZ[x,y]/(x^S - 1, y^m + 1)$,
with coefficients bounded by $mq^{1/m}$.
By definition of $\theta$-representation, we have
$\tilde H(x,\theta) = H(x) \pmod q$,
and similarly for the $G_i$.

After computing the products $\tilde G_i \tilde H$ for $i = 1, \ldots, L/S$,
suppose that
\begin{equation*}
  (\tilde G_i \tilde H)(x,y) = \sum_{j=0}^{S-1} A_{ij}(y) x^j,
   \qquad A_{ij} \in \ZZ[y]/(y^m + 1).
\end{equation*}
Then we have $(G_i H)(x) = (\tilde G_i \tilde H)(x, \theta)
= \sum_j A_{ij}(\theta) x^j \pmod q$ for each $i$.
Therefore, to compute the desired products $G_i H$
with coefficients in $\theta$-representation,
it suffices to apply Proposition \ref{prop:reduction} to each $A_{ij}$,
to compute $\theta$-representations for all of the $A_{ij}(\theta)$.

Let us estimate the cost of the invocations of Proposition \ref{prop:reduction}.
We have $\norm{A_{ij}} \leq S m (m q^{1/m})^2 = S m^3 (q^{1/m})^2$, so
  \[ \lg \norm{A_{ij}} \leq \frac{2 \lg q} m + \lg S + 3 \lg m
                       < \frac{2 \lg q}m + (\lg \lg L)^2 + O(\lg \lg L). \]
From \eqref{eq:m-interval} we have
$\frac{\lg q}m > \frac12(\lg \lg L)^2 \lg \lg \lg L$, so for large $L$,
\begin{equation}
\label{eq:normA-bound}
 \lg \norm{A_{ij}} < \left(2 + \frac{3}{\lg \lg \lg L}\right) \frac{\lg q}{m}.
\end{equation}
The cost of applying Proposition \ref{prop:reduction} for all $A_{ij}$ is thus
  \[ O\left((L/S) S \left\lceil \frac{m \lg \norm{A_{ij}}}{\lg q} \right\rceil \Mss(\lg q)\right)
   = O(L \Mss(\lg q)) = O(L (\lg \lg L)^2 \lg q). \qedhere \]
\end{proof}

\step{Reduce to bivariate products over $\ZZ/q'\ZZ$.}
Let $p'$ be the smallest prime such that $p' = 1 \pmod S$;
by Linnik's theorem we have $p' = O(S^{5.2})$.
Put $q' := (p')^{\alpha'}$ where
\begin{equation*}
  \alpha' := \left\lceil \left(2 + \frac{4}{\lg \lg \lg L}\right) \frac{\lg q}m \bigg\slash \lg \lfloor p'/2 \rfloor \right\rceil.
\end{equation*}
We have the following bounds for $q'$.
\begin{lem}
\label{lem:qprime-bound}
Let $A_{ij}$ be as in the proof of Lemma \ref{lem:step-bivariate},
for $i = 1, \ldots, L/S$ and $j = 0, \ldots, S-1$.
Then $q' \geq 4 \norm{A_{ij}}$ and
 \[ \lg q' < \left(2 + \frac{O(1)}{\lg \lg \lg L}\right) \frac{\lg q}{m}. \]
\end{lem}
\begin{proof}
In what follows, we frequently use the fact that
$\frac{\lg q}{m} \asymp (\lg \lg L)^2 \lg \lg \lg L$
(see~\eqref{eq:m-interval}).
Now, observe that 
$\log_2 q' = \alpha' \log_2 p' \geq \alpha' \lg\lfloor p'/2 \rfloor$,
so by \eqref{eq:normA-bound},
 \[ \log_2 q' \geq \left(2 + \frac{4}{\lg \lg \lg L}\right) \frac{\lg q}m
              \geq \left(2 + \frac{3}{\lg \lg \lg L}\right) \frac{\lg q}m + 2
              \geq \lg \norm{A_{ij}} + 2. \]
Thus $q' \geq 4 \norm{A_{ij}}$.
For the other direction, since
$\lg p' \asymp \lg S = (\lg \lg L)^2$, we have
\begin{multline*}
 \lg q' \leq \alpha' \lg p' \leq \left[\frac{\left(2 + \frac{4}{\lg \lg \lg L}\right) \frac{\lg q}m}{\lg \lfloor p'/2 \rfloor} + 1\right] \lg p'
 < \left(2 + \frac{O(1)}{\lg \lg \lg L}\right) \frac{\lg q}{m} \cdot \frac{\lg p'}{\lg \lfloor p'/2 \rfloor},
\end{multline*}
and $\lg p' / \lg \lfloor p'/2 \rfloor < 1 + O(1)/\lg p' < 1 + O(1)/(\lg \lg L)^2$.
\end{proof}

Now, given as input polynomials
$g_1, \ldots, g_{L/S}, h \in (\ZZ/q'\ZZ)[x,y]/(x^S - 1, y^m + 1)$,
let $\MMbivariateprime(L,q)$ denote the cost of computing the products
$g_1 h, \ldots, g_{L/S} h$, where all input and output coefficients
in $\ZZ/q'\ZZ$ are in standard representation.
\begin{lem}
\label{lem:step-qprime}
We have $\MMbivariate(L,q) < \MMbivariateprime(L,q) + O(L \lg q)$.
\end{lem}
\begin{proof}
We may locate $p'$ by testing $S+1, 2S+1, \ldots$,
in $S^{O(1)} = 2^{O((\lg \lg L)^2)} = O(L)$ bit operations,
and we may easily compute $\alpha'$ and $q'$ within the same time bound.
Now, given as input 
$\tilde G_1, \ldots, \tilde G_{L/S}, \tilde H \in \ZZ[x,y]/(x^S - 1, y^m + 1)$,
we first convert them (in linear time) to polynomials
$g_1, \ldots, g_{L/S}, h \in (\ZZ/q'\ZZ)[x,y]/(x^S - 1, y^m + 1)$,
and then multiply them in the latter ring.
The bound $q' \geq 4 \norm{A_{ij}}$ in Lemma \ref{lem:qprime-bound}
shows that the products over $\ZZ$ may be
unambiguously recovered from those over $\ZZ/q'\ZZ$;
again, this lifting can be done in linear time.
\end{proof}

\step{reduce to DFTs over $\ZZ/q'\ZZ$.}
In this step we will call \textsc{Transform} recursively to handle
certain transforms of length $S$ over $\ZZ/q'\ZZ$.
To check that these calls are permissible,
we must verify the precondition corresponding to \eqref{eq:lgq-bound},
namely $\lg S \leq \lg q' \leq 3 \lg S \lg \lg S$.
The first inequality is clear since $q' \geq p' > S$.
The second inequality follows from \eqref{eq:m-interval},
Lemma \ref{lem:qprime-bound}, and the observation that
$\lg S \lg \lg S \geq (\lg \lg L)^2 \lg \lg \lg L$.

\begin{lem}
\label{lem:step-dft}
We have $\MMbivariateprime(L, q) < \left(\frac{2L}{S} + 1\right) m \TT(S, q') + O(L (\lg \lg L)^2 \lg q)$.
\end{lem}
\begin{proof}
We start by computing various data needed for the recursive calls.
We may compute a primitive $S$-th root of unity in $\ZZ/p'\ZZ$ in
$(p')^{O(1)} = O(L)$ bit operations,
and then Hensel lift it to a principal $S$-th root of unity
$\zeta' \in\ZZ/q'\ZZ$ in $(\lg p' \lg q')^{O(1)} = O(L)$ bit operations.
Just as before, we define $m'$ to be the unique power of two in the interval
\begin{equation}
\label{eq:mprime-interval}
 \frac{\lg q'}{(\lg \lg S)^2 \lg \lg \lg S}
    \leq m'
    < \frac{2 \lg q'}{(\lg \lg S)^2 \lg \lg \lg S},
\end{equation}
and set $\theta' := (\zeta')^{S/2m'}$.
Using Lemmas~\ref{lem:compute-P}, \ref{lem:compute-J},
and~\ref{lem:compute-D}, we may compute $\PP(q', m', \theta')$ in
$(q')^{1+o(1)} = 2^{O((\lg \lg L)^2 \lg \lg \lg L)} = O(L)$ bit operations.

Now suppose we wish to compute the products
$g_1 h, \ldots, g_{L/S} h$, for polynomials
$g_1, \ldots, g_{L/S}, h \in (\ZZ/q'\ZZ)[x,y]/(x^S - 1, y^m + 1)$.
We use the following algorithm.

First we use \textsc{Transform} to transform all $L/S + 1$ polynomials
with respect to $x$, that is,
we compute $g_i((\zeta')^j, y)$ and $h((\zeta')^j, y)$ as elements of
$(\ZZ/q'\ZZ)[y]/(y^m + 1)$, for $i = 1, \ldots, L/S$ and $j = 0, \ldots, S-1$.
Since \textsc{Transform} must be applied separately to every
coefficient $1, y, \ldots, y^{m-1}$,
the total number of calls is $(L/S + 1) m$.
Accessing the coefficient of each $y^k$ also implies
a number of array transpositions whose total cost is
$O((L/S) S m \lg m \lg q') = O(L \lg \lg L \lg q)$.

Next we compute the $(L/S) S = L$ pointwise products
$g_i((\zeta')^j, y) h((\zeta')^j, y)$.
Using Kronecker substitution, each such product in $(\ZZ/q'\ZZ)[y]/(y^m + 1)$
costs $O(\Mss(\lg q))$ bit operations,
as $m(\lg q' + \lg m) = O(\lg q)$.

Finally, we perform $(L/S) m$ inverse transforms with respect to $x$.
It is well known that these may be computed by the same algorithm as the
forward transform, with $\zeta'$ replaced by $(\zeta')^{-1}$,
followed by a division by $S$.
The division may be accomplished by simply multiplying through by
$S^{-1} \pmod{q'}$;
this certainly costs no more than the pointwise multiplication step.
\end{proof}

\begin{cor}
\label{cor:main}
We have $\TT(L, q) < \frac{\lg L}{(\lg \lg L)^2} \left(\frac{2L}{S} + 1\right) m \TT(S, q') + O(L \lg L \lg q)$.
\end{cor}
\begin{proof}
This follows immediately by chaining together Lemmas~\ref{lem:step-convert},
\ref{lem:step-short}, \ref{lem:step-convolution}, \ref{lem:step-bivariate},
\ref{lem:step-qprime}, and~\ref{lem:step-dft}.
\end{proof}

Define
 \[ \TT(L) := \max_q \frac{\TT(L, q)}{L \lg L \lg q}, \]
where the maximum is taken over
all prime powers $q$ satisfying \eqref{eq:lgq-bound}.
(For large $L$, at least one such $q$ always exists.
For example, take $\alpha := 1$ and take $q = p$ to be the smallest prime
satisfying $p = 1 \pmod L$; then Linnik's theorem implies that
\eqref{eq:lgq-bound} holds for this $q$.)

\begin{prop}
\label{prop:recurrence}
We have $\TT(L) < \left(4 + \frac{O(1)}{\lg \lg \lg L}\right) \TT(2^{(\lg \lg L)^2}) + O(1)$.
\end{prop}
\begin{proof}
Dividing the bound in Corollary \ref{cor:main} by $L \lg L \lg q$ yields
 \[ \frac{\TT(L, q)}{L \lg L \lg q} < \left(2 + \frac{S}{L}\right) \frac{m \lg q'}{\lg q} \cdot \frac{\TT(S, q')}{S \lg S \lg q'} + O(1). \]
Applying Lemma \ref{lem:qprime-bound}
and the estimate $S/L < O(1)/\lg \lg \lg L$ yields
 \[ \frac{\TT(L, q)}{L \lg L \lg q} < \left(4 + \frac{O(1)}{\lg \lg \lg L}\right) \TT(S) + O(1). \]
Taking the maximum over allowable $q$ yields the desired bound.
\end{proof}

\begin{cor}
\label{cor:main-bound}
We have $\TT(L) = O(4^{\log^* L})$.
\end{cor}
\begin{proof}
This follows by applying the ``master theorem'' \cite[Prop.~8]{HvdHL-mul}
to the recurrence in Proposition \ref{prop:recurrence}.
Alternatively, it follows by the same method used to deduce
\cite[Cor.~3]{HvdH-vanilla} from \cite[Prop.~2]{HvdH-vanilla}.
The key point is that $2^{(\lg \lg L)^2}$ is dominated by
a~``logarithmically slow'' function of $L$,
such as $\Phi(x) := 2^{(\log \log x)^3}$
(see \cite[\S5]{HvdHL-mul}).
\end{proof}

\begin{rem}
When working with $\theta$-representations, it is possible to multiply
an element of $\ZZ/q\ZZ$ by any power of $\theta$ in linear time,
by simply permuting the coefficients.
In other words, we have available ``fast roots of unity''
in the sense of F\"urer.
Notice however that the algorithm presented in this section
makes no use of this fact!

This raises the question of whether one can design an integer
multiplication algorithm that uses these fast roots in the same way
as in F\"urer's original algorithm, instead of our appeal to Bluestein's trick.
This is indeed possible, and one does obtain a bound of the form
$O(n \lg n \, K^{\log ^* n})$.
In this algorithm, instead of the running time being dominated by the
short transforms, it is dominated by the twiddle factor multiplications,
just as in F\"urer's algorithm.
Unfortunately, this leads to a worse value of $K$, because of the implied
constant in Proposition \ref{prop:arithmetic}.
\end{rem}

\section{Integer multiplication: the top level}
\label{sec:top}

The only complication in building an integer multiplication algorithm
on top of the \textsc{Transform} routine is ensuring that the precomputations
do not dominate the complexity.
We achieve this by means of a multivariate Kronecker-style splitting, as follows.

\begin{proof}
[Proof of Theorem \ref{thm:main}]
Suppose that we wish to compute the product of two $n$-bit integers $u$ and $v$,
for some sufficiently large $n$.
Let $b := \lg n$ and $t := \lceil n/b \rceil^{1/6}$, so that $t^6 b \geq n$.
Decompose $u$ into $t^6$ chunks of $b$ bits,
say $u = u_0 + u_1 2^b + \cdots + u_{t^6 - 1} 2^{(t^6-1)b}$
where $0 \leq u_i < 2^b$ for each $i$, and similarly for~$v$.
Let
 \[ U(x_0, \ldots, x_5) := \sum_{i_0=0}^{t-1} \cdots \sum_{i_5=0}^{t-1} u_{i_0 + t i_1 + \cdots + t^5 i_5} x_0^{i_0} \cdots x_5^{i_5} \in \ZZ[x_0, \ldots, x_5], \]
so that $u = U(2^b, 2^{tb}, \ldots, 2^{t^5 b})$,
and define $V(x_0, \ldots, x_5)$ similarly.
The product $UV$ has degree less than $2t$ in each variable,
so at most $64 t^6$ terms altogether,
and its coefficients are bounded by $2^{2b} n \leq 4n^3$.
We may therefore reconstruct $uv$ from $UV$ using a straightforward
overlap-add procedure
(essentially, evaluating at $(2^b, 2^{tb}, \ldots, 2^{t^5 b})$)
in $O(t^6 \lg n) = O(n)$ bit operations.

Now we consider the computation of $UV$.
Let $L$ be the unique power of two in the interval $2t \leq L < 4t$;
then it suffices to compute the product $UV$ in the ring
$\ZZ[x_0, \ldots, x_5]/(x_0^L - 1, \ldots, x_5^L - 1)$.

For $i = 0, \ldots, 18$, let $q_i$ be the least prime such that
$q_i = 1 \pmod L$ and $q_i = i \pmod{19}$.
Then the $q_i$ are distinct, and by Linnik's theorem they satisfy
$q_i = O(L^{5.2}) = O(t^{5.2}) = O(n^{0.9})$,
so we may locate the $q_i$ in $n^{0.9 + o(1)}$ bit operations,
and they certainly satisfy \eqref{eq:lgq-bound}.
Moreover, for large $n$ we have
$q_0 \cdots q_{18} > L^{19} > 2^{19} t^{19} \geq 2^{19} (n / \lg n)^{19/6} > 4n^3$,
so to compute $UV$ it suffices to compute $UV \pmod{q_i}$ for each $i$
and then reconstruct $UV$ by the Chinese remainder theorem.
The cost of this reconstruction is
$(\lg n)^{1+o(1)}$ bit operations per coefficient,
so $(n / \lg n) (\lg n)^{1+o(1)} = n (\lg n)^{o(1)}$ altogether.

We have therefore reduced to the problem of computing a product in the ring
$(\ZZ/q_i\ZZ)[x_0, \ldots, x_5]/(x_0^L - 1, \ldots, x_5^L - 1)$
for each $i = 0, \ldots, 19$.
To do this, we use \textsc{Transform} to perform forward DFTs of length $L$
with respect to a suitable primitive $L$-th root of unity in $\ZZ/q_i\ZZ$,
for each variable $x_0, \ldots, x_5$ successively;
then we multiply pointwise in $\ZZ/q_i\ZZ$;
finally we perform inverse DFTs and scale the results.
The necessary precomputations for each prime
(finding a suitable root of unity and computing the appropriate
tuple $\PP(q_i, m_i, \theta_i)$) require only $q_i^{1+o(1)} = n^{0.9+o(1)}$
bit operation per prime.
The total cost of the pointwise multiplications is $n (\lg n)^{o(1)}$.
The total number of calls to \textsc{Transform} for each prime is $12 L^5$,
so by Corollary \ref{cor:main-bound} we obtain
\begin{align*}
 \Mint(n) & = O(L^5 \textstyle \sum_{i=0}^{18} \TT(L, q_i)) + n (\lg n)^{o(1)} \\
          & = O(L^6 \textstyle \sum_{i=0}^{18} \TT(L) \lg L \lg q_i) + n (\lg n)^{o(1)} \\
          & = O((n/\lg n) 4^{\log^* L} \lg n \lg n) + n (\lg n)^{o(1)} \\
          & = O(n \lg n \, 4^{\log^* n}).  \qedhere 
\end{align*}
\end{proof}


\bibliography{latticemul}

\end{document}